\documentclass[runningheads,a4paper]{llncs}
\usepackage{amsmath}
\usepackage{amssymb}
\usepackage{graphicx}

\usepackage{url}
\urldef{\mailsa}\path|masaki@cs.e.yamagata-u.ac.jp|
\newcommand{\keywords}[1]{\par\addvspace\baselineskip
\noindent\keywordname\enspace\ignorespaces#1}

\newtheorem{problemI}{Problem I}

\newtheorem{condition}{Condition}

\newcommand{\evend}{\stackrel{e}{\sim}}
\newcommand{\oddd}{\stackrel{o}{\sim}}
\newcommand{\nc}{|\!\!c}
\newcommand{\dol}{\$}
\newcommand{\ket}[1]{\left|#1\right>}

\newcommand{\vectorx}{\mbox{\boldmath $x$}}

\begin{document}

\mainmatter  

\title{Quantum Pushdown Automata with a Garbage Tape}

\titlerunning{Quantum Pushdown Automata with a Garbage Tape}

%
%
\author{Masaki Nakanishi}%
%
\authorrunning{Masaki Nakanishi}

\institute{Faculty of Education, Art and Science, Yamagata University,\\
Yamagata, 990--8560, Japan\\
\mailsa\\
}

%
%

\toctitle{Quantum Pushdown Automata with a Garbage Tape}
\tocauthor{Masaki Nakanishi}
\maketitle

\begin{abstract}
Several kinds of quantum pushdown automaton models have been proposed,
and their computational power is investigated intensively. However,
for some quantum pushdown automaton models, it is not known whether
quantum models are at least as powerful as classical counterparts or
not. This is due to the reversibility restriction. In this paper, we
introduce a new quantum pushdown automaton model that has a garbage
tape. This model can overcome the reversibility restriction by
exploiting the garbage tape to store popped symbols. We show that
the proposed model can simulate any quantum pushdown automaton with a
classical stack as well as any probabilistic pushdown automaton. We
also show that our model can solve a certain promise problem exactly while
deterministic pushdown automata cannot. These results imply that our
model is strictly more powerful than classical counterparts in the
setting of exact, one-sided error and non-deterministic computation.
\end{abstract}

\section{Introduction}
  One important question in quantum computing is whether  a
  computational gap exists between  models that are allowed to  use quantum
  effects and 
  models that are not. Several types of quantum computation models
  have been proposed, including quantum finite automata, quantum counter automata, and quantum
  pushdown automata. 
   Quantum finite automata are the simplest model of 
   quantum computation, and have been investigated 
   intensively\cite{amb-qfa-strength,amb-2qfa_classical,amb-exact-qfa,bro-qfa,cia-qfa,hir-qfa1,hir-qfa2,kon-qfa,moo-q_automata,pas-qfa,cem-qca,yak-realtime_qa,yak-amplification,yak-succinctness,yak-unbounded-error-qfa}. 
 Several quantum automata augmented with additional computational resources have also been proposed, including
  quantum counter automata and quantum pushdown automata 
  \cite{bon-counter,gol-qpa,kra-counter,moo-q_automata,mur-dqpa,naka-qcpda,cem-qca,yam-kcounter,yam-counter}.

  It might be a surprising result that some of simple quantum computation models can be less powerful than classical
  counterparts\cite{kon-qfa,yam-kcounter,yam-counter} due to the reversibility restriction.
    Thus, it is a natural question what kinds of restrictions
  make quantum models less powerful than classical counterparts, and what kinds of computational resources
  make quantum models more powerful. Motivated by those questions, quantum pushdown automata have been investigated.
  Quantum pushdown automata were first proposed in \cite{moo-q_automata}, but their model is 
  the generalized quantum pushdown automata whose evolution does not have to be unitary. 
   Then Golovkins proposed quantum pushdown automata including unitarity criteria\cite{gol-qpa}, and he 
  showed that 
  quantum pushdown automata can recognize every regular language and
  some non-context-free languages. 
  However, it is still open whether 
  Golovkins's model of quantum pushdown automata are more powerful than probabilistic
  pushdown automata or not. In \cite{mur-dqpa}, it is shown that a certain promise problem can be
  solved exactly by Golovkins's model of quantum pushdown automata while it cannot be solved by deterministic
  pushdown automata. However, it is not known whether Golovkins's model can simulate any
  deterministic pushdown automaton or not. This is because quantum computation models must be reversible while
  pop operation deletes the stack-top symbol, which is not a reversible operation.
  In \cite{naka-qcpda}, a quantum pushdown automaton model that has a classical stack is proposed, and it is shown that
  the model is  strictly more powerful than classical counterparts in the setting of one-sided error as well as
  non-deterministic computation. 
  
  The above mentioned results are for the models whose state transitions are described by unitary operators.
  It is known that by allowing more general operators such as trace preserving completely positive (TPCP) maps, quantum finite automata
  can simulate classical counterparts as well as several quantum finite automata mentioned above\cite{hir-qfa1,hir-qfa2}.
  These results were generalized and it was shown how to define general quantum operators for other models in \cite{yak-unbounded-error-qfa}.
  For counter automata and pushdown automata, it is also known that generalized quantum models (i.e., the models
  that can use TPCP maps) can simulate classical counterparts\cite{cem-qca,yak-realtime_qa}.

  In this paper, we focus on the restricted quantum computation models (i.e., the models whose state transitions are described by unitary operators) 
  rather than the general models (i.e., the models whose state transitions are described by TPCP maps). 
  As mentioned above, it is known that the generalized quantum computation models can
  simulate classical counterparts and sometimes can be strictly more powerful than classical counterparts. 
  Nevertheless, studying restricted models is important. That is, our goal is to investigate what kinds of restriction makes quantum 
  models less powerful and under what kinds of restrictions quantum models are still more powerful than classical counterparts.
  This could lead to understand the source of the power of quantum computation in architecturally restricted models such as quantum
  automata.
 
    Motivated by these discussions, we introduce a new model of quantum pushdown automata, called quantum pushdown automata with
  a garbage tape. This model has a garbage tape on which popped symbols are stored, and thus,
  we can pop the stack-top symbol preserving reversibility. The garbage tape is a write-only memory,
  and thus, classical pushdown automata cannot exploit it. 
  A quantum computation model that has a write-only memory was proposed in \cite{yak-write-only-2}. The
  model uses a write-only memory in order to control interference between distinct computation paths.
  In our model, the write-only garbage tape is restricted to store popped symbols.
  Also the similar notion of garbage tapes were proposed in \cite{cia-qfa,pas-qfa}. In those models, a garbage tape
  is used to make transitions reversible. Our model is constructed so as to take advantages of both of a write-only
   tape and a garbage tape.

  Another motivation is that 
  it is expected that investigating quantum pushdown automata reveals how last-in first-out manner of memory access
  affects (or limits) quantum computation. However, for this purpose, Golovkins's model\cite{gol-qpa} is too restrictive on pop operation, i.e.,
  we can pop a stack-top symbol only if we can delete stack-top symbol preserving reversibility.
  Thus, we cannot identify from which the impossibilities come from, reversibility or last-in first-out 
  manner of memory access.
  In contrast, our model is useful for this purpose since pop operations can always be executed
  preserving reversibility.
  
  In this paper, we show that the proposed model
can simulate any quantum pushdown automaton with a classical stack, which is 
proposed in \cite{naka-qcpda}, as well as any classical pushdown automaton.
It is known that quantum pushdown automata with a classical stack are strictly more powerful than classical
counterparts in the setting of one-sided error and non-deterministic computation\cite{naka-qcpda}. Thus, so is our model.
We also show that our model can solve a certain promise problem (Problem~I) exactly while deterministic 
pushdown automata cannot. This implies that our model is strictly more powerful than
classical counterparts also in the setting of exact computation. It is a common technique to apply the pumping lemma (or 
Ogden's lemma\cite{ogden}, which is a generalization of the pumping lemma) in order to show that a language is not context-free,
i.e., pushdown automata cannot recognize the language.
However, our problem is a promise problem. Thus, we cannot apply the pumping lemma to our case.\footnote{As far
as the author knows, \cite{mur-dqpa} is
the only exception in which the pumping lemma is used for a promise problem. The technique in \cite{mur-dqpa} can be
applied only to the limited cases.} 
In \cite{ama-pumping}, the pumping lemma is proved through the analysis of pushdown automata. 
We modify their notion of {\it full state},
and use it to show the impossibility by directly analyzing time evolution of pushdown automata. This is a new technique
to prove that a certain promise problem cannot be computed by pushdown automata.
For OBDD models, an impossibility proof for a certain partial function, 
which is a function counterpart of promise problems, 
was shown recently in \cite{abl-qbdd}.


 This paper is organized as follows: In Sect.~2, we define quantum
pushdown automata with a garbage tape. In Sect.~3, we show how to simulate
classical pushdown automata and quantum pushdown automata with a classical stack by quantum pushdown automata with
a garbage tape. In Sect.~4, we show there is a promise problem that quantum pushdown automata
with a garbage tape can solve exactly while deterministic pushdown automata cannot. In Sect.~5, we discuss comparison between
quantum pushdown automata with and without a garbage tape.

\section{Preliminaries}
 A quantum pushdown automaton with a garbage tape (QPAG)  has an   
 input tape, a stack 
 and a garbage tape.
A QPAG also has a finite state control. 
 The input tape contains a classical input string, and its tape head is implemented by qubits
 that represent the position of the tape head. The stack and the garbage tape are 
 implemented by qubits that represent contents of the stack and the garbage tape, respectively.
  The finite state control is also implemented by qubits that represents the current state.
  A QPAG
 reads the stack top symbol and  the input symbol pointed by the input tape-head, and then
 evolves
 as follows: The tape head can move to the right or stay at the same position, 
 the finite state control moves to the next state, and a stack symbol is pushed to the stack
 or popped from the stack. When we pop a symbol from the stack, the popped symbol is
 written on the garbage tape, moving the garbage tape head to the right. This allows a QPAG
 to pop the stack top symbol preserving reversibility.
 We define QPAGs 
 formally as follows. 
 \begin{definition}
  A quantum pushdown automaton with a garbage tape (QPAG)
  is defined 
  as the following 7-tuple:
$
  M=(Q, \Sigma,\Gamma, \delta, q_0, Q_{acc}, Q_{rej}),   
$
  where $Q$ is a set of states, $\Sigma$ is a set of input
  symbols including the  
  left and the right endmarkers $\{\nc, \dol \}$, respectively, $\Gamma$
  is a set of stack symbols including the bottom symbol $Z$, $\delta$
  is a quantum state transition 
  function ($\delta:(Q\times\Sigma\times\Gamma\times Q\times G\cup\{\varepsilon, pop\} \times \{0, 1\} \}) 
  \longrightarrow \bbbc$), where $G(\subseteq  (\Gamma\setminus \{Z\})^+)$ is a finite set and
  $(\Gamma \setminus \{Z\})^+$ is the set of all nonempty strings of finite length from alphabet $\Gamma \setminus \{Z\}$, 
  $q_0$ is the initial 
  state, 
  $Q_{acc}$ $(\subseteq Q)$ 
  is the set of accepting states, and $Q_{rej}$ $(\subseteq Q)$ is
  the set of 
  rejecting states, where $Q_{acc}\cap Q_{rej}=
  \emptyset$. 
  \qed
 \end{definition}
 
 $\delta(q,a,b,q', b', D)=\alpha$ means that the amplitude of the transition from $q$ 
 to $q'$ updating the input tape head to $D$ ($D = 1$ means `right' and $D
= 0$ means
 `stay') and pushing $b'$ to the stack (or popping the stack-top symbol if $b'=pop$) is $\alpha$ when  
 reading input symbol $a$ and  stack
 symbol $b$. A configuration of
 a QPAG  is 
 $(q, k, w_{s}, w_{g})$, where 
 $q\in Q$ is the current state of the finite state control,
  $k$ is the position of the input tape head, and 
  $w_{s}$ and $w_{g}$ are the strings on the stack and the garbage tape, respectively. 
 We store a configuration of a QPAG in a quantum register, 
 where a basis  state is described as $\ket{q, k, 
  w_{s}, w_{g}}$.
  For  input string $\vectorx$, we define the time
  evolution operator 
  $U^{\vectorx}$ as follows: 
\begin{eqnarray*}
 \lefteqn{U^{\vectorx}(\ket{q, k, w_{s}b, w_{g}})}&&\\
&=&\sum_{q'\in Q, b'\in G\cup\{\varepsilon, pop\}, D\in\{0,1\}}  \delta(q,x(k),b,q',b', D)
\ket{q', k+D, w_{s}', w_{g}'},
\end{eqnarray*}
  where $x(k)$ is the $k$-th input symbol of input $\vectorx$, 
\begin{eqnarray*}
w_{s}'& =& \left\{ \begin{array}{ll}w_{s}bb'  & \mbox{ if } b' \neq pop\\
                                          w_{s} & \mbox{ if } b' = pop \end{array}\right.\\
 \mbox{and } w_{g}' &=& \left\{ \begin{array}{ll}w_{g} &\mbox{ if } b' \neq pop\\
                                         w_{g}b &\mbox{ if } b' = pop 
                                         \mbox{ ($b$ is the popped stack-top symbol).}
                                         \end{array}\right.
\end{eqnarray*}
  If $U^{\vectorx}$ is unitary (for any 
input string $\vectorx$), 
  then the corresponding QPAG is 
  well-formed. A well-formed QPAG is considered valid in
  terms of the quantum 
  theory. We consider only well-formed QPAGs.
  Let the initial quantum state and the initial position of the input tape head be
  $q_0$ and `0', respectively. We define $\ket{\psi_0}$ as 
  $\ket{\psi_0}=\ket{q_0,0, Z, \varepsilon}$. 
  We also define $E_{non}$, $E_{acc}$ and
  $E_{rej}$ as follows: 
\begin{eqnarray*} 
   E_{non}&= &span\{\ket{q,k,w_s, w_g}| q\not\in Q_{acc} \mbox{ and } q\not\in Q_{rej} \},\\
   E_{acc}&=& span\{\ket{q,k, w_s, w_g}| q\in Q_{acc}\}, \   E_{rej}= span\{\ket{q,k, w_s, w_g}| q\in Q_{rej}\}.
\end{eqnarray*}
  We define observable $\cal O$ as ${\cal O}=E_{non}\oplus E_{acc}\oplus E_{rej}$.
  For notational simplicity, we 
  define the outcome of a measurement corresponding to $E_j$ as $j$ for $j\in\{non, acc, rej\}$.  
  A QPAG computation proceeds as follows: 

  \begin{description}
  \item[(a)] $U^{\vectorx}$ is applied to $\ket{\psi_i}$, and we obtain 
	 $\ket{\psi_{i+1}}=U^{\vectorx}\ket{\psi_i}$.
  \item[(b)]  $\ket{\psi_{i+1}}$ is measured with respect to 
	 ${\cal O}$. Let $\ket{\phi_j}$ be the projection of
	     $\ket{\psi_{i+1}}$ to $E_j$. Then each outcome $j$ is
	     obtained with 
	     probability $|\ket{\phi_j}|^2$. Note that this measurement causes
	     $\ket{\psi_{i+1}}$ to collapse to
	     $\frac{1}{|\ket{\phi_j}|}\ket{\phi_j}$, where $j$ is 
	     the obtained outcome.
  \item[(c)] If the outcome of the measurement is $acc$ or $rej$, the automaton
        outputs the measurement result and halts. Otherwise, go to (a).
  \end{description}

To check well-formedness, we show the following theorem.
\begin{theorem}
\label{theorem:well-formed}
A QPAG $M$ is well-formed if the quantum state transition function of $M$ 
satisfies the following conditions.
\begin{enumerate}
\item $\forall (q, a, b)\in Q\times \Sigma \times \Gamma$, 
\[\Sigma_{q'\in Q, b'\in G\cup\{\varepsilon,pop\}, D\in\{0,1\}} 
|\delta(q,a,b, q',b',D)|^2 = 1,\]

\item $\forall (q_1, a, b)\neq (q_2, a, b) \in Q\times \Sigma \times \Gamma$, 
\[\Sigma_{q'\in Q, b'\in G\cup\{\varepsilon,pop\}, D\in\{0,1\}} 
\delta^*(q_1,a,b, q',b',D) \delta(q_2,a,b, q',b',D) = 0,\]

\item 
\begin{enumerate}
\item $\forall (q_1, a, b_1)\neq (q_2, a, b_2)\in Q\times \Sigma \times \Gamma$, 
$\forall b_3\in G\cup\{\varepsilon\}$,
\[\Sigma_{q'\in Q, D\in\{0,1\}} 
\delta^*(q_1,a,b_1, q',\varepsilon,D) \delta(q_2,a,b_2, q',b_3 b_1,D) = 0,\]

\item $\forall (q_1, a, b_1)\neq (q_2, a, b_2)\in Q\times \Sigma \times \Gamma$, 
$\forall b_3\in G\cup\{\varepsilon\}$,
\[\Sigma_{q'\in Q, D\in\{0,1\}} 
\delta^*(q_1,a,b_1, q',pop ,D) \delta(q_2,a,b_2, q',b_3,D) = 0,\]
\end{enumerate}

\item $\forall (q_1, a_1, b)\neq (q_2, a_2, b)\in Q\times \Sigma \times \Gamma$,
\[\Sigma_{q'\in Q, b\in G\cup\{\varepsilon, pop\}}
\delta^*(q_1,a_1,b, q',b',0) \delta(q_2,a_2,b, q',b',1) = 0,\]

\item 
\begin{enumerate}
\item $\forall (q_1, a_1, b_1)\neq (q_2, a_2, b_2)\in Q\times \Sigma \times \Gamma$, 
$\forall D_1\neq D_2 \in\{0,1\}$, 
$\forall b_3\in G\cup\{\varepsilon\}$, 
\[\Sigma_{q'\in Q}
\delta^*(q_1,a_1,b_1, q',\varepsilon, D_1) \delta(q_2,a_2,b_2, q', b_3 b_1, D_2) = 0,\]

\item $\forall (q_1, a_1, b_1)\neq (q_2, a_2, b_2)\in Q\times \Sigma \times \Gamma$, 
$\forall D_1\neq D_2 \in \{0,1\}$, 
$\forall b_3\in G\cup \{\varepsilon\}$, 
\[\Sigma_{q'\in Q}
\delta^*(q_1,a_1,b_1, q',pop , D_1) \delta(q_2,a_2,b_2, q', b_3, D_2) = 0,\]
\end{enumerate}

\end{enumerate}
\end{theorem}
\begin{proof}
The matrix $U^{\vectorx}$ is unitary if and only if 
the columns of $U^{\vectorx}$ are orthonormal. 
The condition (1) implies that each column of $U^{\vectorx}$ is normalized.
The rest of the conditions implies any two distinct columns are orthogonal.
The condition (2) is for the columns corresponding to the configurations
in which only the state is different. 
The conditions (3-a) and (3-b) are for the columns corresponding to
the configurations in which 
the position of the tape-head is the same but the stack contents are different.
The condition (4) is for the columns corresponding to the configurations in which
the position of the tape-head is different but the stack contents are the same.
The conditions (5-a) and (5-b) are for the columns corresponding to the
configurations in which the position of the tape head is different and the
stack contents are also different.
Note that in the case of (3-b) and (5-b), the contents of 
the garbage tape are different between the two configurations; one is shorter
than the other by one symbol. Also note that in the case of the rest, the contents
of the garbage tape are the same between the two configurations.
\qed
\end{proof}

\section{Simulation of QCPDAs}
In this section, we show that a QPAG can
simulate a quantum pushdown automaton with a classical stack (QCPDA).
Since QCPDAs can simulate any probabilistic pushdown automata\cite{naka-qcpda}, QPAGs can
simulate any probabilistic pushdown automata as well.
We describe the definitions of probabilistic pushdown automata
 and QCPDAs in Appendices A and B, respectively, or readers may refer to \cite{naka-qcpda}. 
A quantum pushdown automaton with a classical stack(QCPDA) is a quantum pushdown automaton
whose classical stack operations are determined by measurement results. 
We can use
the garbage tape so that if we measure the garbage tape, the stack contents will be identical
among all the basis states contained in the resulting superposition. Therefore, we can simulate a QCPDA
by a QPAG.
\begin{theorem}
\label{theorem:qcpda_simulation}
Let $M_{qc}=(Q, \Sigma, \Gamma, \delta, q_0, \sigma, Q_{acc}, Q_{rej})$ be a QCPDA. Then, there exists a QPAG $M_q$
such that for any input, the acceptance probability of $M_q$ is the same as
that of $M_{qc}$.
\end{theorem}
\begin{proof}
For a transition of $M_{qc}$  from state $q$ to $q'$ moving the input tape head
to $D$, we construct the corresponding transitions of $M_q$, which consist of three
successive transitions, as follows: 
Note that the stack operation of $M_{qc}$ is determined solely by the state $q'$ to which it transits,
 denoted by $\sigma(q')$.
We add two new states $q_a$ and $q_b$ to $Q$ and also add $\sigma(q')$ to $\Gamma$. 
Then, we replace the original transition
with the transition from $q$ to $q_a$ such that the stack operation is the same as the original transition ($\sigma(q')$), 
the direction of the tape head is $D$ and the transition probability is also the same. We define the transition
from $q_a$ to $q_b$, whose probability is one, as a transition pushing the label $\sigma(q')$ to the stack, the input tape head
staying at the same position. We also define the transition from $q_b$ to $q'$, whose probability is one, as
a transition popping $\sigma(q')$ from the stack and moving it to the garbage tape, the input tape head
staying at the same position. This records the history of stack operations in the garbage tape.
Thus, if the history of stack operations are different between two computation paths, they do not
interfere with each other since the contents of the garbage tape are different. This means that
if we measure the garbage tape, 
the contents of the stack are identical between any basis states contained in the resulting superposition at any moment
of computation. In other words, if we trace out the garbage tape, then, the stack configuration is not
in a superposition but in a classical mixture of basis states. Thus, it can be regarded as a classical stack,
and the resulting QPAG $M_q$ simulates the original QCPDA $M_{qc}$.
\qed
\end{proof}

It is known that QCPDAs can recognize a certain non-context-free language with one-sided error\cite{naka-qcpda}. 
This means that
QPAGs are strictly more powerful than classical pushdown automata in the setting of
one-sided error as well as non-deterministic computation.
\begin{corollary}
The class of languages recognized by one-sided error QPAGs properly
includes the class of languages recognized by one-sided error probabilistic pushdown automaton as well as
by non-deterministic pushdown automaton. \qed
\end{corollary}

\section{Possibility and Impossibility of Solving a Certain Promise Problem}
We say that two strings, $u$ and $v$, have even (resp. odd) distinctions, denoted by $u \stackrel{e}{\sim} v$ (resp. 
$u \stackrel{o}{\sim} v$), if $|u| = |v|$ and 
$u$ and $v$ are different at even (resp. odd) number of positions.
For example, $1100\evend 1111$ since the third and the fourth bits are different between the two strings, and
$1000\oddd 1111$ since the second, the third and the fourth bits are different between the two strings.
We define a promise problem, Problem I, as follows:

\begin{problemI}
\ 
The input is promised to be of the form $w_1\# w_2 \# w_3$, where $w_1, w_2\in \{a,b,c\}^n$ and $w_3\in\{a,b,c,d\}^n$.
\begin{description}
\item[Yes-instances] are formed by the strings $w_1 \# w_2 \# w_3$ such that
\[((w_1 \stackrel{e}{\sim} w_2^R) \mbox{ xor } (w_1 \stackrel{e}{\sim} w_3^R)) = 1.\]
\item[No-instances] are formed by the strings $w_1 \# w_2 \# w_3$ such that
\[((w_1 \stackrel{e}{\sim} w_2^R) \mbox{ xor } (w_1 \stackrel{e}{\sim} w_3^R)) = 0.\]
\qed
\end{description}
\end{problemI}


We show that QPAGs can solve Problem I exactly while deterministic pushdown
automata cannot solve it. This result combined with Theorem~\ref{theorem:qcpda_simulation} 
implies that QPAGs are strictly more powerful than
classical pushdown automata in the setting of exact computation.

\begin{theorem}
\label{theorem:qpag_p1}
There exists a QPAG that solves Problem I exactly.
\end{theorem}
\begin{proof}
We use the same technique as in Theorem~3.1 of \cite{mur-dqpa}.
We construct a QPAG, $M$,  that solves Problem I as follows:
We consider two sub-automata, $M_1$ and $M_2$, such that $M_1$ (resp. $M_2$)  computes whether $w_1 \evend w_2^R$
(resp. $w_1 \evend w_3^R$), and run them in a superposition.
It is straightforward to see that $M_1$ and $M_2$ can be implemented by reversible deterministic pushdown automata 
with a garbage tape,
which is a special case of QPAGs,
and we can construct $M_1$ and $M_2$ so that the contents of the garbage tape at the moment of reading 
the right-endmarker can be 
the same between the two sub-automata. Then, we utilize the algorithm in \cite{cle-quantum} (the improved Deutsch-Jozsa algorithm\cite{deutsch-jozsa}) 
to compute the exclusive-or exactly using the two sub-automata
as the oracle for Deutsch's problem\cite{deu-quantum}.
We show the transition function of $M$ in Appendix C.
\qed
\end{proof}
The reason why we can use the same technique as in Theorem~3.1 of \cite{mur-dqpa} even though our model and the model used in \cite{mur-dqpa}
seems incomparable is the following. 
When the stack-top symbol is popped, it is always written in the garbage tape in our model. This makes an 
entanglement between the stack contents and the garbage tape. Sometimes, this can be an unwanted behavior and make
our model weaker than the model in \cite{mur-dqpa}. However, in our algorithm shown in the proof of Theorem~\ref{theorem:qpag_p1},
the contents of the garbage tape at the moment of reading 
the right-endmarker can be 
the same between the two sub-automata. Therefore, the stack contents and the garbage tape are separable at the moment 
of reading the right-endmarker, which
causes no problem when using the same technique in Theorem~3.1 of \cite{mur-dqpa}.

In the following, we show that no deterministic pushdown automata can solve Problem I.
\begin{theorem}
\label{theorem:no_dpa}
No deterministic pushdown automata can solve Problem I.
\qed
\end{theorem}

We introduce several lemmas in order to prove Theorem~\ref{theorem:no_dpa}.
We divide $w_1$ into two segments $w_1 =w_{1L}w_{1R}$.
Similarly, we divide $w_2$ and $w_3$ as $w_2 =w_{2L}w_{2R}$ and $w_3 =w_{3L}w_{3R}$, respectively.
In the following discussion, we assume that there exists a deterministic pushdown 
automaton that solves Problem I.
Let $h_{max}(k)$ be the maximum height
of the stack over all $w_1$'s  at the moment of reading the $k$-th symbol of $w_1$. 
Note that stack height can increase at most $O(1)$ when reading each symbol\footnote{Note that, on the other hand,
stack height may decrease more than $\omega(1)$ when reading each
input symbol.}. 
Then, it is obvious that
there is a constant, $c$, for which the following holds:
\[
h_{max}(\frac{n}{c}) < \log_{|\Gamma|}\left( 3^{\frac{c-1}{c}n}/(\#states\cdot n(n+1))\right), 
\]
 where $\#states$ denotes the number of states of the
finite state control, and $n=|w_1|$.
We fix such a constant $c$, and also fix the length of $w_{1L}$ to be $n/c$.

We say that pushdown automaton $M$  is in a {\it state-configuration} of $(q, a)$ if
$M$ is in the state $q$ and the stack-top symbol is $a$. In other words, a state-configuration 
is a configuration of a pushdown automaton ignoring the position of the tape head and the stack contents
except for the stack-top. The notion of a state-configuration is a modification of the notion of {\it full state}
in \cite{ama-pumping}. 
Note that the tape head can be stationary at a transition. Thus, the stack height can increase
(or decrease) multiple times with multiple transitions during reading one symbol.
Let $h_b(i)$ and $c_b(i)$ denote the stack height and the state-configuration, respectively, 
immediately before reading the $i$-th symbol of the input.
Also let $h_{a}(i)$ denote the set of stack heights that consists of the stack height after reading the $i$-th
symbol and the stack heights during reading the $(i+1)$-th symbol with the
tape head being stationary on the $(i+1)$-th symbol. 
For each $h\in h_{a}(i)$, let $c_{a}(i, h)$ be the corresponding state-configuration.
We define the notations ``$h_b(i)>h_{a}(j)$'' and ``$h_b(i)-h_{a}(j)$'' as follows:
$h_b(i)>h_{a}(j)$ iff  $h_b(i)>min_{h'\in h_{a}(j)}h'$.
$h_b(i)-h_{a}(j)=h_b(i) - min_{h'\in h_{a}(j)}h'$.
A {\it zero-stack} pair is a pair $(l,r)$ $(1\leq l<r\leq n)$ such that
$h_b(l)\in h_{a}(r)$ and $h_b(l) \not> h_{a}(t)$ for any $t$ ($l<t<r$).
Then, we have the following lemma.

\begin{lemma}
\label{lemma:memoryless_substring}
We fix $w_1$ arbitrarily. Let  $(i,j)$ be a zero-stack pair
such that the maximum of $h_b(k)-h_{a}(l)$ for $i<k<l<j$ is $\omega(1)$.
Then, for any zero-stack pair $(i', j')$  ($1\leq i'< j' < i$ or $j < i'<j'\leq n$), 
the maximum of $h_b(k)-h_{a}(l)$ for $i'<k<l<j'$ is $O(1)$.
\end{lemma}
\begin{proof}
We consider  a zero-stack pair  $(i,j)$ $(1\leq i<j\leq n)$ such that
the maximum of $h_b(k)-h_{a}(l)$ for $i<k<l<j$ is $\omega(1)$.
Let the maximum (resp. minimum) height of the stack during processing
from the $i$-th symbol to the $j$-th symbol be $h_{max}$ (resp. $h_{min}$).
Note that $h_{max}-h_{min}>\omega(1)$.
For each $h\in \{h_{min},\ldots, h_{max}\}$, 
let $ZS_h$ be the set of zero-stack pairs such that 
$ZS_h=\{ (l,r) |  h_b(l)=h, i\leq l < r \leq j \}$. Note that for at least a constant fraction of 
$\{h_{min},\ldots, h_{max}\}$, $ZS_h$ is nonempty.
For each $h\in \{h_{min},\ldots, h_{max}\}$, 
we choose at most one $(l_h,r_h)\in ZS_h$ such that $l_h < l_{h+1}$ and $r_{h+1}\leq r_h$.
It is obvious that we can have such $(l_h, r_h)$'s for at least a constant fraction of
$\{h_{min},\ldots, h_{max}\}$.
Let $(c_{a}(k, h), t)$ be a pair of a state-configuration and an input symbol where
$t\in \Sigma$ is the input symbol pointed by the tape head 
at the moment when the automaton is in the state-configuration $c_{a}(k, h)$ with the $k$ and $h$.
Then, there exists two distinct zero-stack pairs $(l_{h_1}, r_{h_1})$ and 
$(l_{h_2}, r_{h_2})$  $(h_1<h_2)$ such that $c_b(l_{h_1}) = c_b(l_{h_2})$ and 
$(c_{a}(r_{h_1}, h_1), t) = (c_{a}(r_{h_2}, h_2), t)$ for some $t$ 
since $|\Sigma|$ and the number of possible state-configurations are both $O(1)$ 
while we have $\omega(1)$ pairs of $(l_h, r_h)$'s.
We divide $w_1$ as $w_1=uvxyz$ where $u=w_1(1)\cdots w_1(l_{h_1}-1)$, 
$v=w_1(l_{h_1})$ $\cdots w_1(l_{h_2}-1)$, 
$x=w_1(l_{h_2})\cdots w_1(r_{h_2})$, 
$y=w_1(r_{h_2}+1)\cdots w_1(r_{h_1})$, and 
$z=w_1(r_{h_1}+1)\cdots w_1(n)$, where $w_1(i)$ denotes the $i$-th symbol of $w_1$.
Then, for any $i\geq 0$, the configuration after reading $uv^ixy^iz$ and the configuration after reading $uvxyz$
are  the same, including the contents of the stack.

We assume that there exists two zero-stack pairs  $(i_1,j_1)$ and $(i_2, j_2)$ 
$(1\leq i_1<j_1<i_2<j_2\leq n)$ such that
the maximum of $h_b(k)-h_{a}(l)$ for $i_1<k<l<j_1$
and the maximum for $i_2<k<l<j_2$
 are both $\omega(1)$.
Then, we can divide $w_1$ in two ways: $w_1=u_k v_k x_k y_k z_k$ with $(i_k,j_k)$ $(k\in\{1,2\})$.
It is obvious that there exist $p$ and $q$ such that $|u_1 v_1^p x_1 y_1^p z_1|
=|u_2 v_2^q x_2 y_2^q z_2|$. Thus, there exist two inputs $u_1 v_1^p x_1 y_1^p z_1$ and
$u_2 v_2^q x_2 y_2^q z_2$ for which the configurations after reading the two inputs
are the same, including the contents of the stack. 
This implies that for any completion of the inputs, both of
$u_1 v_1^p x_1 y_1^p z_1$ and $u_2 v_2^q x_2 y_2^q z_2$ leads to the same answer,
which is a contradiction.
\qed
\end{proof}

Let $w_{pref}$ be a string for which there is a zero-stack pair $(i,j)$ and
the maximum of $h_b(k)-h_{a}(l)$ for $i<k<l<j$ is $\omega(1)$ where $|w_{pref}| = c|w_{1L}|$
for some constant $c$ ($0<c<1$).
If there is no such zero-stack pair for any long enough $w_{pref}$, we define $w_{pref}$ to be an empty string.
We fix such a $w_{pref}$.
We define $a^+=b, b^+=c, c^+=a$.
For two strings $u, v$ $\in \{a, b, c\}^n$, we say $u\leq v$ iff 
$[(u(k)=x) \longrightarrow (v(k)=x$ or $v(k)=x^+)]$, where $x\in\{a,b,c\}$ and $u(k)$ (resp. $v(k)$) represents the $k$-th
symbol of $u$ (resp. $v$).
Let $WL_{all}$ be the set of $w_{1L}$'s such that 
$WL_{all}=\{w_{pref} a^{|w_{1L}|-|w_{pref}|-k} b^k | 0\leq k \leq |w_{1L}|-|w_{pref}|-1\}$
 ($= \{w_{pref}aaa\ldots aaa$, $w_{pref}aaa\ldots aab$, $w_{pref}aaa\ldots abb,$ 
 $w_{pref}aaa\ldots bbb, \ldots, w_{pref}abb\ldots bbb\}$). 
Note that for any two distinct $u, v\in WL$, $u\leq v$ or $v\leq u$.
Then, we have the following lemma.
\begin{lemma}
\label{lemma:WL}
There exists $WL\subseteq WL_{all}$ satisfying the following conditions:
(1) Any $w\in WL$ leads to the same state-configuration, say $C_{WL}$. 
(2) Given a constant $c$, after reading $w$, the stack contents between 
the top and the $c$-th from the top are the same among all $w\in WL$. 
(3) $|WL| = \Theta(n)$. \qed
\end{lemma}
\begin{proof}
There exists  a constant fraction of $WL_{all}$, which is $WL$, 
satisfying the first and the second conditions of
the lemma since the number of possible state-configurations is a constant 
and the number of possible stack contents 
between the top and the $c$-th
from the top is also a constant. 
It is obvious that $|WL|=\Theta(n)$ since $|WL_{all}| = |w_{1L}|-|w_{pref}| = \Theta(n)$.
\qed
\end{proof}

We consider the case that the following Condition I holds:

\begin{condition}
There exists $w_{1L}\in WL$ and $w_{2L}$ such that for at least $1/(n+1)$ fraction of $\{w_{1R}\}$, stack height is less than
$\log_{|\Gamma|}(3^{n-|w_{1L}|}/(\#states\cdot n(n+1)))$ at the moment of reading the last symbol of
$w_{1L}w_{1R}w_{2L}$.
\qed
\end{condition}

In this case, at the moment when stack height is less than  $\log_{|\Gamma|}(3^{n-|w_{1L}|}/$ $(\#states\cdot n(n+1)))$,
the number of possible configurations (including stack contents and the position
of the input tape head) is less than $\frac{1}{n+1}3^{n-|w_{1L}|}$, which means
there exist at least two distinct partial inputs $w_{1L}w_{1R} w_{2L}$ and 
$w_{1L}w_{1R}' w_{2L}$ that result in the same configuration (including stack contents and the position
of an input tape head) since $|\{w_{1R}\}| = 3^{n-|w_{1L}|}$. Thus both of
$w_{1L}w_{1R} w_{2L}$ and $w_{1L}w_{1R}' w_{2L}$
 lead to the same answer for any
completion of the rest of the input. This is a contradiction.
Thus, we can say the negation of Condition I holds. In this case, given $w_2$,
at every step of processing $w_2$, for at most $1/(n+1)$ fraction of $\{w_{1R}\}$, stack height becomes less than
$\log_{|\Gamma|}(3^{n-|w_{1L}|}/(\#states\cdot n(n+1)))$.
Thus, for at most $n/(n+1)$ fraction of $\{w_{1R}\}$, stack height becomes less than
$\log_{|\Gamma|}(3^{n-|w_{1L}|}/(\#states\cdot n(n+1)))$ while processing $w_2$; for at least $1/(n+1)$
fraction of $\{w_{1R}\}$, stack height is always more than or equal to
$\log_{|\Gamma|}(3^{n-|w_{1L}|}/(\#states\cdot n(n+1)))$ while processing $w_2$.
We consider the case that the following Condition II holds:

\begin{condition}
For any $w_{1L}\in WL$ and $w_2$, at least $1/(n+1)$ fraction of $\{w_{1R}\}$, stack height is always
greater than or equal to $\log_{|\Gamma|}(3^{n-|w_{1L}|}/(\#states\cdot n(n+1)))$
 while processing $w_2$.
\qed
\end{condition}

We define $w_{3L}$  as the prefix of $w_3$ such that stack height is always 
higher than $\hat{h}-O(1)$ $(=\hat{h}')$ during reading $w_{1R} w_2 w_{3L}$ and it becomes $\hat{h}'$ when reading
the last symbol of $w_{3L}$, where $\hat{h}$ denotes the stack height after reading the 
last symbol of $w_{1L}$.
If stack height is always higher than $\hat{h}'$ during
reading $w_3$, we define $w_{3L}=w_3$.



\begin{lemma}
\label{lemma:W2}
We assume that there exists a deterministic pushdown automaton that solves Problem I.
Then, there exist $w_{1R}$, $k$ ($1\leq k \leq n$) and a set $W_2$ of $w_2$'s such that, 
starting from $C_{WL}$, $w_{1R} w_2 w_{3L}$ leads to the same state-configuration for all
$w_2\in W_2$ where $w_{3L}=d^k$, stack height is
always greater than or equal to $\hat{h}-O(1)$, and $|W_2|=\Omega(\frac{1}{n^2}3^n)$, where 
$C_{WL}$ is as in Lemma~\ref{lemma:WL}.
\end{lemma}
\begin{proof}
Note that for each $w_2$, there are more than $\frac{1}{n+1}3^{|w_{1R}|}$ of $w_{1R}$'s for which 
stack height is always
greater than or equal to $\log_{|\Gamma|}(3^{n-|w_{1L}|}/(\#states\cdot n(n+1)))$
 while processing $w_2$ by Condition II. This means that for some $w_{1R}$, 
 there are $\Omega(\frac{1}{n}3^n)$ of $w_2$'s  for which 
 stack height is always
greater than or equal to $\log_{|\Gamma|}(3^{n-|w_{1L}|}/(\#states\cdot n(n+1)))$
 while processing $w_2$. By Lemma~\ref{lemma:memoryless_substring} and the fact that 
 $\hat{h} < \log_{|\Gamma|}( 3^{|w_{1R}|}/(\#states\cdot n(n+1)))$, 
 the lemma follows immediately.
\qed
\end{proof}

We fix $w_{1R}$, $k$ and $W_2$ as those in Lemma~\ref{lemma:W2} in the following.
For $WL$ in Lemma~\ref{lemma:WL}, the following lemma holds.
\begin{lemma}
\label{lemma:distinction}
We assume that there exists a deterministic pushdown automaton that solves Problem I.
For $w_{1R}$, $k$ and $W_2$ in Lemma~\ref{lemma:W2}, 
there exist $w_{1L}\in WL$ and two distinct partial inputs 
$w_{1L}w_{1R} w_2 w_{3L}$ and $w_{1L}w_{1R} w_2' w_{3L}$ ($w_2, w_2'\in W_2$  and $w_{3L}=d^k$) 
 such that
$w_1\evend w_2^R$ and $w_1\oddd w_2'^R$.
\end{lemma}
\begin{proof}
Let $WL=\{w_{1L}^1, w_{1L}^2, \ldots, w_{1L}^m\}$ where $w_{1L}^{i}\leq w_{1L}^{i+1}$.
$W^1_{2,even}$ denotes the set of $w_2\in \{a,b,c\}^n$ such that $w_{1L}^1w_{1R} \evend w_2^R$. 
Also $W^2_{2,even}$ denotes the set of $w_2\in W_2^1$ such that 
$w_{1L}^2w_{1R} \evend w_2^R$.
Similarly, $W^i_{2,even}$ denotes the set of $w_2\in W_{2,even}^{i-1}$ such that 
$w_{1L}^iw_{1R} \evend w_2^R$.
In other words, for all $w_2\in W_{2,even}^i$ and $j\leq i$, $w_{1L}^jw_{1R} \evend w_2^R$.
We show that $|W_{2,even}^i| \leq c |W_{2,even}^{i-1}|$ for some constant $c<1$ in the following.
We consider the positions at which $w_{1L}^{i-1}$ and $w_{1L}^i$ differ. 
We define the set of such positions to be $D^i$. Note that $w_{1L}^{i-1}(k) = a$ and $w_{1L}^{i}(k) = b$
for $k\in D^i$, where $w(k)$ represents the $k$-th symbol of $w$.
We define $S=\{ w_2\in W_{2,even}^{i-1} | \mbox{$\exists k\in D^i$ $w_2^R(k) = b$ or $w_2^R(k) = c$.}   \}$,
where $w_2(k)$ denotes the $k$-th symbol of $w_2$.
It is obvious that $|S| \geq c_1|W_{2,even}^{i-1}|$ for some constant $c_1<1$.
For $w_2 \in S$, let $l$ be the largest position in $D^i$ such that $w_2^R(l)=b$ or $w_2^R(l)=c$.
We assume that $w_2^R(l) = b$ without loss of generality. We consider $w_2'$ such that
$w_2'^R(i) = w_2^R(i)$ for $i\neq l$ and $w_2'^R(i) = c$. It is obvious that 
$w_2'$ is also in $S$. Then, either $w_{1L}^iw_{1R} \oddd w_2^R$
or $w_{1L}^iw_{1R} \oddd w_2'^R$ holds. This implies that a half of elements in $S$ cannot belong to
$W_{2,even}^{i}$. Thus, $|W_{2,even}^{i}| \leq |W_{2,even}^{i-1}| - \frac{c_1}{2}|W_{2,even}^{i-1}|
=c_2|W_{2,even}^{i-1}|$, where $c_2 = 1-\frac{c_1}{2}$.
%
%
Similar to $W^1_{2,even}$, we define $W^1_{2,odd}$, and then, similarly, it can be shown
that $|W_{2,odd}^i| \leq c |W_{2,odd}^{i-1}|$ for some constant $c<1$.
Therefore, $|W_{2,even}^i|$ and  $|W_{2,odd}^i|$ can be smaller 
than $|W_2|$ for $i\in \Theta(n)$. The lemma follows.
\qed
\end{proof}

{\noindent \bf  (Proof of Theorem~\ref{theorem:no_dpa})}

We assume that there exists a classical deterministic pushdown automaton
that solves Problem I. Then, by Lemma~\ref{lemma:distinction}, 
we have two input strings, 
$w_a=w_{1L}w_{1R} w_2 w_{3L}$ and $w_b=w_{1L}w_{1R} w_2' w_{3L}$ ($w_2, w_2'\in W_2$  and $w_{3L}=d^k$),
such that
$w_1 \oddd w_2^R$ and $w_1 \evend w_2'^R$.
 We fix $w_{3R}=d^{n-k}$. Then, the answer only depends on the number of distinctions between
 $w_1$ and $w_2^R$ (or $w_2'^R$). Thus, one is YES and the other is NO for $w_a$ and $w_b$.
 However, the configurations (including the contents of the stack and the position
 of the input tape head) at the moment of reading the last symbol of $w_{3L}$ are the
 same between $w_a$ and $w_b$ if $k\neq n$. On the other hand,
 if $k=n$, the state-configuration at the moment of reading the last symbol of
 $w_a$ and $w_b$ are the same.
 Thus, both of $w_a$ and $w_b$ lead to the same answer. This is a contradiction.
 \qed






\section{Comparison between Quantum Pushdown Automata with and without a Garbage Tape -- Concluding Remarks}
In this paper, we showed that QPAGs are strictly more powerful than classical pushdown automata in the setting
of exact, one-sided error  and nondeterministic computation.
In this section, we discuss comparison between quantum pushdown automata with and without 
a garbage tape. Our conjecture is that Problem I cannot be solved exactly  by quantum pushdown automata
without a garbage tape, which is Golovkins's model\cite{gol-qpa}, since it seems to be impossible to compute
$w_1 \evend w_2^R$ or $w_1 \evend w_3^R$ without a garbage tape. On the other hand, in the QPAG model, 
popped symbols are always stored in the garbage tape. Thus, if the contents of the garbage tape are different
between two computation paths, they no longer interfere with each other. In other words, only the two computation
paths that have the same contents in the garbage tape can interfere with each other. This might make the QPAG model
less powerful than Golovkins's model. Therefore, we conjecture that the class of languages recognized by the two models
are incomparable. We also conjecture that even the generalized quantum pushdown automata without a garbage tape
constructed by the technique in \cite{yak-unbounded-error-qfa} cannot solve Problem I.
At least, the generalized model of quantum pushdown automata without a garbage tape cannot execute the algorithm in 
Theorem~\ref{theorem:qpag_p1}. This is because, although the garbage tape is in a superposition in the middle
of the computation of the algorithm, the generalized quantum pushdown automaton cannot represent such a superposition without
a garbage tape. Thus, our model and the generalized model without a garbage tape might also be incomparable.

\subsubsection*{Acknowledgments}
This work was supported  by JSPS KAKENHI Grant Nos. 24500003 and 24106009.

\bibliographystyle{splncs03}
\bibliography{ref.bib}

\begin{thebibliography}{10}
\providecommand{\url}[1]{\texttt{#1}}
\providecommand{\urlprefix}{URL }

\bibitem{abl-qbdd}
Ablayev, F., Gainutdinova, A., Khadiev, K., Yakary{\i}lmaz, A.: Very narrow
  quantum {OBDD}s and width hierarchies for classical {OBDD}s. In: Proceedings
  of 16th International Workshop on Descriptional Complexity of Formal Systems
  (DCFS'14). pp. 53--64 (2014)

\bibitem{ama-pumping}
Amarilli, A., Jeanmougin, M.: A proof of the pumping lemma for context-free
  languages through pushdown automata (2012), coRR, abs/1207.2819

\bibitem{amb-qfa-strength}
Ambainis, A., Freivalds, R.: 1-way quantum finite automata: strengths, weakness
  and generalizations. In: Proceedings of the 29th Symposium on Foundations of
  Computer Science (FOCS'98). pp. 332--341 (1998)

\bibitem{amb-2qfa_classical}
Ambainis, A., Watrous, J.: Two-way finite automata with quantum and classical
  states. Theoretical Computer Science  287(1),  299--311 (2002)

\bibitem{amb-exact-qfa}
Ambainis, A., Yakary{\i}lmaz, A.: Superiority of exact quantum automata for
  promise problems. Information Processing Letters  112(7),  289--291 (2012)

\bibitem{bon-counter}
Bonner, R., Freivalds, R., Kravtsev, M.: Quantum versus probabilistic one-way
  finite automata with counter. In: Proceedings of the 28th Conference on
  Current Trends in Theory and Practice of Informatics (SOFSEM2001). pp.
  181--190 (2001)

\bibitem{bro-qfa}
Brodsky, A., Pippenger, N.: Characterizations of 1-way quantum finite automata.
  SIAM Journal on Computing  31(5),  1456--1478 (2002)

\bibitem{cia-qfa}
Ciamarra, M.P.: Quantum reversibility and a new model of quantum automaton. In:
  Proceedings of the 13th International Symposium on Fundamentals of
  Computation Theory (FCT'01). pp. 376--379 (2001)

\bibitem{cle-quantum}
Cleve, R., Ekert, A., Macchiavello, C., Mosca, M.: Quantum algorithms
  revisited. Proceedings of the Royal Society A  454,  339--354 (1998)

\bibitem{deu-quantum}
Deutsch, D.: The {C}hurch-{T}uring principle and the universal quantum
  computer. Proceedings of the Royal Society A  400,  97--117 (1985)

\bibitem{deutsch-jozsa}
Deutsch, D., Jozsa, R.: Rapid solution of problem by quantum computation.
  Proceedings of the Royal Society A  439,  553--558 (1992)

\bibitem{gol-qpa}
Golovkins, M.: Quantum pushdown automata. In: Proceedings of 27th Conference on
  Current Trends in Theory and Practice of Informatics (SOFSEM2000). pp.
  336--346 (2000)

\bibitem{hir-qfa1}
Hirvensalo, M.: Various aspects of finite quantum automata. In: Proceedings of
  Developments in Language Theory 2008 (DLT2008). pp. 21--33 (2008)

\bibitem{hir-qfa2}
Hirvensalo, M.: Quantum automata with open time evolution. International
  Journal of Natural Computing Research (IJNCR)  1(1),  70--85 (2010)

\bibitem{kon-qfa}
Kondacs, A., Watrous, J.: On the power of quantum finite state automata. In:
  Proceedings of the 38th Symposium on Foundations of Computer Science
  (FOCS'97). pp. 66--75 (1997)

\bibitem{kra-counter}
Kravtsev, M.: Quantum finite one-counter automata. In: Proceedings of 26th
  Conference on Current Trends in Theory and Practice of Informatics
  (SOFSEM1999). pp. 432--442 (1999)

\bibitem{moo-q_automata}
Moore, C., Crutchfield, J.P.: Quantum automata and quantum grammars.
  Theoretical Computer Science  237(1--2),  275--306 (2000)

\bibitem{mur-dqpa}
Murakami, Y., Nakanishi, M., Yamashita, S., Watanabe, K.: Quantum versus
  classical pushdown automata in exact computation. IPSJ Journal  46(10),
  2471--2480 (2005)

\bibitem{naka-qcpda}
Nakanishi, M., Hamaguchi, K., Kashiwabara, T.: Expressive power of quantum
  pushdown automata with classical stack operations under the perfect-soundness
  condition. IEICE Transactions on Information and Systems  E89-D(3),
  1120--1127 (2006)

\bibitem{ogden}
Ogden, W.: A helpful result for proving inherent ambiguity. Mathematical
  Systems Theory  2(3),  191 -- 194 (1968)

\bibitem{pas-qfa}
Paschen, K.: Quantum finite automata using ancilla qubits (2000), technical
  report, University of Karlsruhe, available at
  http:{\slash}{\slash}digbib.ubka.uni-karlsruhe.de{\slash}volltexte{\slash}1452000

\bibitem{cem-qca}
Say, A.C.C., Yakary{\i}lmaz, A.: Quantum counter automata. International
  Journal of Foundations of Computer Science  23(5),  1099--1116 (2012)

\bibitem{yak-realtime_qa}
Yakary{\i}lmaz, A.: Superiority of one-way and realtime quantum machines. RAIRO
  - Theoretical Informatics and Applications  46(04),  615--641 (2012)

\bibitem{yak-write-only-2}
Yakary{\i}lmaz, A., Freivalds, R., Say, A.C.C., Agadzanyan, R.: Quantum
  computation with wirte-ony memory. Natural Computing  11(1),  81--94 (2012)

\bibitem{yak-amplification}
Yakary{\i}lmaz, A., Say, A.C.C.: Efficient probability amplification in two-way
  quantum finite automata. Theoretical Computer Science  410(20),  1932--1941
  (2009)

\bibitem{yak-succinctness}
Yakary{\i}lmaz, A., Say, A.C.C.: Succinctness of two-way probabilistic and
  quantum finite automata. Discrete Mathematics and Theoretical Computer
  Science  12(4),  19--40 (2010)

\bibitem{yak-unbounded-error-qfa}
Yakary{\i}lmaz, A., Say, A.C.C.: Unbounded-error quantum computation with small
  space bounds. Information and Computation  209(6),  873--892 (2011)

\bibitem{yam-kcounter}
Yamasaki, T., Kobayashi, H., Imai, H.: Quantum versus deterministic counter
  automata. Theoretical Computer Science  334(1--3),  275--297 (2005)

\bibitem{yam-counter}
Yamasaki, T., Kobayashi, H., Tokunaga, Y., Imai, H.: One-way probabilistic
  reversible and quantum one-counter automata. Theoretical Computer Science
  289(2),  963--976 (2002)

\end{thebibliography}

\newpage

\section*{Appendix A: Probabilistic Pushdown Automata}
\begin{definition}
  A probabilistic pushdown automaton (PPA) is defined
  as the following 7-tuple:
\[
  M=(Q, \Sigma,\Gamma, \delta, q_0,  Q_{acc}, Q_{rej}),   
\]
 where $Q$ is a set of states, $\Sigma$ is a set of input
  symbols including the 
  left and the right endmarkers $\{\nc, \dol\}$, respectively, $\Gamma$
  is a set of stack symbols including the bottom symbol $Z$, 
   $\delta$
  is a state transition 
  function ($\delta:(Q\times\Sigma\times\Gamma\times Q\times G\cup \{\varepsilon,pop\}\times \{0,1\} ) 
  \longrightarrow [0,1]$),   where $G$ ($\subseteq (\Gamma\setminus \{Z\})^+$) is a finite set and $(\Gamma\setminus\{Z\})^+$ 
  is the set of all nonempty strings of finite length from alphabet $\Gamma\setminus\{Z\}$, 
  $q_0$ is the initial 
  state,  
  $Q_{acc} (\subseteq Q)$ 
  is the set of accepting states, and $Q_{rej} (\subseteq Q)$ is
  the set of 
  rejecting states, where $Q_{acc}\cap Q_{rej}=
  \emptyset$. 
  We restrict that for all $q, q', a, D$, $\delta(q,a,Z,q',pop, D)=0$. 
  \qed
\end{definition}

 $\delta(q,a,b,q',w, D)=\alpha$ means that the probability of the transition
 from $q$ 
 to $q'$ moving the head to $D$ with stack operation $w$ is $\alpha$ when  
 reading  input symbol $a$ and  stack
 symbol $b$. Note that for each input symbol and each stack symbol, the sum
 of the weights (i.e. the probabilities) of outgoing 
 transitions of a state must be 1. 
 Computation halts when it enters the accepting or rejecting
 states.

\section*{Appendix B: Quantum Pushdown Automata with a Classical Stack}
 A quantum  pushdown automata
 with a classical 
 stack (QCPDA) has an   
 input tape to which a quantum head is attached and 
 a classical stack to which a classical stack top pointer is
 attached. 
A QCPDA also has a quantum finite state control.
 The quantum finite state control 
 reads the stack top symbol pointed by the classical stack top
 pointer and  the input symbol pointed by the quantum head. Stack
 operations are  
 determined  solely by the results of measurements of the quantum
 finite state
 control.
  We define QCPDAs 
 formally as follows. 
 \begin{definition}%
  A quantum pushdown automaton with a classical stack (QC\-PDA)
  is defined 
  as the following 8-tuple:
\[
  M=(Q, \Sigma,\Gamma, \delta, q_0, \sigma, Q_{acc}, Q_{rej}),   
\]
  where $Q$ is a set of states, $\Sigma$ is a set of input
  symbols including the  
  left and the right endmarkers $\{\nc, \dol \}$, respectively, $\Gamma$
  is a set of stack symbols including the bottom symbol $Z$, $\delta$
  is a quantum state transition 
  function ($\delta:(Q\times\Sigma\times\Gamma\times Q\times \{0,1\}) 
  \longrightarrow \bbbc$),   
  $q_0$ is the initial 
  state, 
   $\sigma$ is a function by which
  stack operations are determined ($\sigma: Q\setminus(Q_{acc}\cup Q_{rej}) \longrightarrow G\cup \{\varepsilon,pop\}$), where $G$ ($\subseteq (\Gamma\setminus\{Z\})^+$) is 
  a finite set and $(\Gamma\setminus\{Z\})^+$ is the set of all nonempty strings 
  of finite length from alphabet $\Gamma\setminus\{Z\}$,  
  $Q_{acc}$ $(\subseteq Q)$ 
  is the set of accepting states, and $Q_{rej}$ $(\subseteq Q)$ is
  the set of 
  rejecting states, where $Q_{acc}\cap Q_{rej}=
  \emptyset$. 
  We restrict that for all $q, q', a, D$, if $\sigma(q')=pop$, then $\delta(q,a,Z,q',D)=0$.
  \qed
 \end{definition}
 
 $\delta(q,a,b,q',D)=\alpha$ means that the amplitude of the transition from $q$ 
 to $q'$ moving the quantum head to $D$ ($D = 1$ means `right' and $D
= 0$ means
 `stay') is $\alpha$ when  
 reading input symbol $a$ and  stack
 symbol $b$. The configuration of
 the quantum portion of a QCPDA  is a pair  
 $(q,k)$, where $k$ is the position of the quantum head and $q$ is in
 $Q$. It is 
 obvious that the number of configurations of the quantum portion is
 $n|Q|$, where $n$ is the input length. 

 A superposition of the configurations of the quantum portion of a QCPDA is any element of
 $l_2(Q\times\bbbz_n)$ of unit 
 length, where $\bbbz_n=\{0,1,\ldots,n-1\}$. For each
 configuration, we define a column vector 
 $\ket{q,k}$ as follows:
 \begin{itemize}
  \item $\ket{q,k}$ is an $n|Q|\times 1$ column vector.  
  \item The row corresponding to ($q,k$) is 1, and the other rows are 0.
 \end{itemize}
  For  input word $\vectorx$ (i.e., the string on the input tape
between  $\nc$ and $\dol$) and  stack symbol $a$, we define a time
  evolution operator 
  $U^{\vectorx}_{a}$ as follows: 
\[
U^{\vectorx}_{a}(\ket{q,k})
=\sum_{q'\in Q,D\in\{0,1\}}  \delta(q,x(k),a,q',D)\ket{q',k+D},
\]
  where $x(k)$ is the $k$-th input symbol of input $\vectorx$.
  If $U^{\vectorx}_{a}$ is unitary (for any $a \in \Gamma$ and for any 
input word $\vectorx$), that is,
  $U^{\vectorx}_{a}U^{\vectorx\dag}_{a}=U^{\vectorx\dag}_{a} 
  U^{\vectorx}_{a}=I$, where $U_a^{\vectorx\dag}$ is the transpose conjugate of
  $U_a^{\vectorx}$, 
  then the corresponding QCPDA is 
  well-formed. A well-formed QCPDA is considered valid in
  terms of the quantum 
  theory. We consider only well-formed QCPDAs.

We describe how the quantum portion and
the classical stack of a QCPDA work in the following.

  Let the initial quantum state and the initial position of the head be
  $q_0$ and `0', respectively. We define $\ket{\psi_0}$ as 
  $\ket{\psi_0}=\ket{q_0,0}$. 
  We also define $E_{w}$, $E_{acc}$ and
  $E_{rej}$ as follows: 
\begin{eqnarray*} 
   E_{w}&= &span\{\ket{q,k}| \sigma(q)=w \},\\
   E_{acc}&=& span\{\ket{q,k}| q\in Q_{acc}\},\\
   E_{rej}&=& span\{\ket{q,k}| q\in Q_{rej}\}.
\end{eqnarray*}
  We define observable $\cal O$ as ${\cal O}=\oplus_j E_j$, where
  $j$ is `acc', `rej' or $w\in G\cup \{\varepsilon,pop\}$. For notational simplicity, we 
  define the outcome of a measurement corresponding to $E_j$ as $j$.  

  A QCPDA computation proceeds as follows: 

{\noindent\underline{  For input word $\vectorx$, the quantum portion works as follows:}}

  \begin{description}
  \item[(a)] $U^{\vectorx}_{a}$ is applied to $\ket{\psi_i}$. Let
	 $\ket{\psi_{i+1}}=U^{\vectorx}_{a}\ket{\psi_i}$, where $a$ is
	     the stack top symbol.
  \item[(b)]  $\ket{\psi_{i+1}}$ is measured with respect to the observable
	 ${\cal O}=\oplus_j E_j$. Let $\ket{\phi_j}$ be the projection of
	     $\ket{\psi_{i+1}}$ to $E_j$. Then each outcome $j$ is
	     obtained with 
	     probability $|\ket{\phi_j}|^2$. Note that this measurement causes
	     $\ket{\psi_{i+1}}$ to collapse to
	     $\frac{1}{|\ket{\phi_j}|}\ket{\phi_j}$, where $j$ is 
	     the obtained outcome. Then go to (c).
  \end{description}
\noindent\underline{The classical stack
works as follows:} 
  \begin{description}
   \item[(c)] Let the outcome of the
	      measurement  be $j$. If $j$ is `acc' (resp. `rej') then it outputs `accept'
	      (resp. `reject'), and the computation halts. If $j$ is `$\varepsilon$', then the stack is
	      unchanged. If $j$ is `$pop$', then the stack top symbol is
	      popped. Otherwise ($j$ is a word in $G$ in
	      this case), word $j$ is pushed. Then, go to (a) and repeat.
  \end{description}


\section*{Appendix C: State Transition Function of the QPAG that Solves Problem I}
We describe the state transition function of the QPAG $M=(Q, \Sigma, \Gamma,
\delta, q_0, Q_{acc},$ $Q_{rej})$ that solves Problem I in the following,
where $Q=\{q_0\} \cup \{q_i^{I,j}, q_i^{I,j} | i\in\{1, 2\}, j\in\{0, 1\} \} \cup 
\{q_f^{acc}, q_f^{rej}, q_f^{-,0}, q_f^{-,1}\}$, $\Sigma=\{a,b,c,d,\#,\nc,\dol\}$, 
$\Gamma=\{a,b,$ $c,Z\}$, the initial state is $q_0$, $Q_{acc}=\{q_f^{acc}\}$, and $Q_{rej}=\{q_f^{rej}\}$.
Note that sub-automaton $M_1$ (resp. $M_2$) consists of the states $\{q_0, q_1^{I,0}, q_1^{I,1}, q_1^{O,0}, q_1^{O,1}\}$
(resp. $\{q_0, q_2^{I,0}, q_2^{I,1},$ $q_2^{O,0}, q_2^{O,1}\}$).
We first describe the outline of the behavior of $M$.
$M$ consists of the following three stages:
\begin{description}
\item[Stage I] $M$ pushes $w_1$ into the stack.
\item[Stage II] $M_1$ and $M_2$ run in a superposition. 
\begin{description}
\item [$M_1$] runs in a superposition of states $q_1^{I,0}$ and $q_1^{I,1}$. 
 $M_1$ reads $w_2$ and pops the stack-top symbol
one by one. If the input symbol is different from the stack-top symbol, the current state
changes from $q_1^{I,x}$ to $q_1^{I,x\oplus 1}$ ($x\in\{0,1\}$), otherwise $M_1$ stays at the same state.
Then, $M_1$ skips $w_3$.
\end{description}
\begin{description}
\item [$M_2$] runs in a superposition of states $q_2^{I,0}$ and $q_2^{I,1}$. 
 $M_2$ skips $w_2$. Then, $M_2$ reads $w_3$ and pops the stack-top symbol
one by one. If the input symbol is different from the stack-top symbol, the current state
changes from $q_2^{O,x}$ to $q_2^{O,x\oplus 1}$ ($x\in\{0,1\}$), otherwise $M_2$ stays at the same state.
\end{description}
\item[Stage III] $M$ reads the right-endmarker, and then, Hadamard transform is applied to $M$.
\end{description}

We describe the state transition function below.
In the following, * denotes a wild card, which matches any of $a, b, c, Z$.

{\noindent \bf Stage I}
\begin{flalign*}
\begin{array}{lll}
\delta(q_0, \nc, Z, q_0, \varepsilon, 1) = 1, &&\\
\delta(q_0, a, *, q_0, a, 1) = 1,&
\delta(q_0, b, *, q_0, b, 1) = 1,&
\delta(q_0, c, *, q_0, c, 1) = 1,\\
\delta(q_0, \#, *, q_1^{I,0}, \varepsilon, 1) = \frac{1}{2},&
\delta(q_0, \#, *, q_1^{I,1}, \varepsilon, 1) = -\frac{1}{2},&\\
\delta(q_0, \#, *, q_2^{I,0}, \varepsilon, 1) = \frac{1}{2},&
\delta(q_0, \#, *, q_2^{I,1}, \varepsilon, 1) = -\frac{1}{2}
\end{array}
\end{flalign*}

{\noindent \bf Stage II}
\begin{flalign*}
\begin{array}{lll}
\delta(q_1^{I,0}, a, a, q_1^{I,0}, pop, 1) = 1,&
\delta(q_1^{I,0}, b, b, q_1^{I,0}, pop, 1) = 1,&
\delta(q_1^{I,0}, c, c, q_1^{I,0}, pop, 1) = 1,
\end{array}&&
\end{flalign*}
\begin{flalign*}
\begin{array}{lll}
\\
\delta(q_1^{I,0}, a, b, q_1^{I,1}, pop, 1) = 1,&
\delta(q_1^{I,0}, a, c, q_1^{I,1}, pop, 1) = 1,&\\
\delta(q_1^{I,0}, b, a, q_1^{I,1}, pop, 1) = 1,&
\delta(q_1^{I,0}, b, c, q_1^{I,1}, pop, 1) = 1,&\\
\delta(q_1^{I,0}, c, a, q_1^{I,1}, pop, 1) = 1,&
\delta(q_1^{I,0}, c, b, q_1^{I,1}, pop, 1) = 1,&\\
\delta(q_1^{I,0}, \#, *, q_1^{O,0}, \varepsilon, 1) = 1,
\end{array}&&
\end{flalign*}
\begin{flalign*}
\begin{array}{lll}
\delta(q_1^{I,1}, a, a, q_1^{I,1}, pop, 1) = 1,&
\delta(q_1^{I,1}, b, b, q_1^{I,1}, pop, 1) = 1,&
\delta(q_1^{I,1}, c, c, q_1^{I,1}, pop, 1) = 1,\\
\delta(q_1^{I,1}, a, b, q_1^{I,0}, pop, 1) = 1,&
\delta(q_1^{I,1}, a, c, q_1^{I,0}, pop, 1) = 1,&\\
\delta(q_1^{I,1}, b, a, q_1^{I,0}, pop, 1) = 1,&
\delta(q_1^{I,1}, b, c, q_1^{I,0}, pop, 1) = 1,&\\
\delta(q_1^{I,1}, c, a, q_1^{I,0}, pop, 1) = 1,&
\delta(q_1^{I,1}, c, b, q_1^{I,0}, pop, 1) = 1,&\\
\delta(q_1^{I,1}, \#, *, q_1^{O,1}, \varepsilon, 1) = 1,\\
\\
\delta(q_1^{O,0}, a, Z, q_1^{O,0}, \varepsilon, 1) = 1,&
\delta(q_1^{O,0}, b, Z, q_1^{O,0}, \varepsilon, 1) = 1,&\\
\delta(q_1^{O,0}, c, Z, q_1^{O,0}, \varepsilon, 1) = 1,&
\delta(q_1^{O,0}, d, Z, q_1^{O,0}, \varepsilon, 1) = 1,
\end{array}&&
\end{flalign*}
\begin{flalign*}
\begin{array}{lll}
\delta(q_1^{O,1}, a, Z, q_1^{O,1}, \varepsilon, 1) = 1,&
\delta(q_1^{O,1}, b, Z, q_1^{O,1}, \varepsilon, 1) = 1,&\\
\delta(q_1^{O,1}, c, Z, q_1^{O,1}, \varepsilon, 1) = 1,&
\delta(q_1^{O,1}, d, Z, q_1^{O,1}, \varepsilon, 1) = 1,&\\
\\
\delta(q_2^{I,0}, a, *, q_2^{I,0}, \varepsilon, 1) = 1,&
\delta(q_2^{I,0}, b, *, q_2^{I,0}, \varepsilon, 1) = 1,&
\delta(q_2^{I,0}, c, *, q_2^{I,0}, \varepsilon, 1) = 1,\\
\delta(q_2^{I,0}, \#, *, q_2^{O,0}, \varepsilon, 1) = 1,\\
\\
\delta(q_2^{I,1}, a, *, q_2^{I,1}, \varepsilon, 1) = 1,&
\delta(q_2^{I,1}, b, *, q_2^{I,1}, \varepsilon, 1) = 1,&
\delta(q_2^{I,1}, c, *, q_2^{I,1}, \varepsilon, 1) = 1,\\
\delta(q_2^{I,1}, \#, *, q_2^{O,1}, \varepsilon, 1) = 1,\\
\\
\delta(q_2^{O,0}, a, a, q_2^{O,0}, pop, 1) = 1,&
\delta(q_2^{O,0}, b, b, q_2^{O,0}, pop, 1) = 1,&
\delta(q_2^{O,0}, c, c, q_2^{O,0}, pop, 1) = 1,\\
\delta(q_2^{O,0}, a, b, q_2^{O,1}, pop, 1) = 1,&
\delta(q_2^{O,0}, a, c, q_2^{O,1}, pop, 1) = 1,&
\delta(q_2^{O,0}, a, d, q_2^{O,1}, pop, 1) = 1,\\
\delta(q_2^{O,0}, b, a, q_2^{O,1}, pop, 1) = 1,&
\delta(q_2^{O,0}, b, c, q_2^{O,1}, pop, 1) = 1,&
\delta(q_2^{O,0}, b, d, q_2^{O,1}, pop, 1) = 1,\\
\delta(q_2^{O,0}, c, a, q_2^{O,1}, pop, 1) = 1,&
\delta(q_2^{O,0}, c, b, q_2^{O,1}, pop, 1) = 1,&
\delta(q_2^{O,0}, c, d, q_2^{O,1}, pop, 1) = 1,\\
\\
\delta(q_2^{O,1}, a, a, q_2^{O,1}, pop, 1) = 1,&
\delta(q_2^{O,1}, b, b, q_2^{O,1}, pop, 1) = 1,&
\delta(q_2^{O,1}, c, c, q_2^{O,1}, pop, 1) = 1,\\
\delta(q_2^{O,1}, a, b, q_2^{O,0}, pop, 1) = 1,&
\delta(q_2^{O,1}, a, c, q_2^{O,0}, pop, 1) = 1,&
\delta(q_2^{O,1}, a, d, q_2^{O,0}, pop, 1) = 1,\\
\delta(q_2^{O,1}, b, a, q_2^{O,0}, pop, 1) = 1,&
\delta(q_2^{O,1}, b, c, q_2^{O,0}, pop, 1) = 1,&
\delta(q_2^{O,1}, b, d, q_2^{O,0}, pop, 1) = 1,\\
\delta(q_2^{O,1}, c, a, q_2^{O,0}, pop, 1) = 1,&
\delta(q_2^{O,1}, c, b, q_2^{O,0}, pop, 1) = 1,&
\delta(q_2^{O,1}, c, d, q_2^{O,0}, pop, 1) = 1,
\end{array}&&
\end{flalign*}

{\noindent \bf Stage III}
\begin{flalign*}
\begin{array}{lll}
\delta(q_1^{O,0}, \dol, Z, q_f^{-,0}, \varepsilon, 1) = \frac{1}{2},&
\delta(q_1^{O,0}, \dol, Z, q_f^{acc}, \varepsilon, 1) = \frac{1}{2},&\\
\delta(q_1^{O,0}, \dol, Z, q_f^{-,1}, \varepsilon, 1) = \frac{1}{2},&
\delta(q_1^{O,0}, \dol, Z, q_f^{rej}, \varepsilon, 1) = \frac{1}{2},\\
\end{array}&&
\end{flalign*}
\begin{flalign*}
\begin{array}{lll}
\delta(q_1^{O,1}, \dol, Z, q_f^{-,0}, \varepsilon, 1) = \frac{1}{2},&
\delta(q_1^{O,1}, \dol, Z, q_f^{acc}, \varepsilon, 1) = -\frac{1}{2},&\\
\delta(q_1^{O,1}, \dol, Z, q_f^{-,1}, \varepsilon, 1) = \frac{1}{2},&
\delta(q_1^{O,1}, \dol, Z, q_f^{rej}, \varepsilon, 1) = -\frac{1}{2},\\
\end{array}&&
\end{flalign*}
\begin{flalign*}
\begin{array}{lll}
\delta(q_2^{O,0}, \dol, Z, q_f^{-,0}, \varepsilon, 1) = -\frac{1}{2},&
\delta(q_2^{O,0}, \dol, Z, q_f^{acc}, \varepsilon, 1) = -\frac{1}{2},&\\
\delta(q_2^{O,0}, \dol, Z, q_f^{-,1}, \varepsilon, 1) = \frac{1}{2},&
\delta(q_2^{O,0}, \dol, Z, q_f^{rej}, \varepsilon, 1) = \frac{1}{2},\\
\\
\delta(q_2^{O,1}, \dol, Z, q_f^{-,0}, \varepsilon, 1) = -\frac{1}{2},&
\delta(q_2^{O,1}, \dol, Z, q_f^{acc}, \varepsilon, 1) = \frac{1}{2},&\\
\delta(q_2^{O,1}, \dol, Z, q_f^{-,1}, \varepsilon, 1) = \frac{1}{2},&
\delta(q_2^{O,1}, \dol, Z, q_f^{rej}, \varepsilon, 1) = -\frac{1}{2},
\end{array}&&
\end{flalign*}
It is straightforward to see that the corresponding time evolution operator can be extended
to be unitary.

\medskip

{\bf \noindent Remark}

The reason why we can use the same technique as in Theorem~3.1 of \cite{mur-dqpa} even though our model and the model used in \cite{mur-dqpa}
seems incomparable is the following. 
When the stack-top symbol is popped, it is always written in the garbage tape in our model. This makes an 
entanglement between the stack contents and the garbage tape. Sometimes, this can be an unwanted behavior and make
our model weaker than the model in \cite{mur-dqpa}. However, in our algorithm shown in the proof of Theorem~\ref{theorem:qpag_p1},
the contents of the garbage tape at the moment of reading 
the right-endmarker can be 
the same between the two sub-automata. Therefore, the stack contents and the garbage tape are separable at the moment 
of reading the right-endmarker, which
causes no problem when using the same technique in Theorem~3.1 of \cite{mur-dqpa}.

\end{document}